%% file: root.tex
\def\useieeelayout{0}
\def\showall{0}

\newcommand{\inConf}[1]{\if\useieeelayout1{#1}\fi\if\showall1{\color{green!50!black}In ACC: #1}\fi}
\newcommand{\inArxiv}[1]{\if\useieeelayout0{#1}\else\if\showall1{\color{blue}In ArXiV: #1}\fi\fi}

\if\useieeelayout1
\documentclass[letterpaper, 10 pt, conference]{ieeeconf}  
\IEEEoverridecommandlockouts                              %
\usepackage{graphicx}
\usepackage{amsmath}
\usepackage{optidef}
 
\usepackage{amsthm}
\usepackage{amssymb}
\usepackage{amsfonts}       
\usepackage{multicol}
\usepackage{cuted}
\usepackage{multirow}
\usepackage{caption, subcaption}
\newtheorem{theorem}{Theorem}
\newtheorem{lemma}{Lemma}
\newtheorem{proposition}{Proposition}
\newtheorem{definition}{Definition}
\newtheorem{problem}{Problem}
\newtheorem{conjecture}[theorem]{Conjecture}
\newtheorem{corollary}[theorem]{Corollary}
\usepackage{xcolor}
\usepackage{cite}
\usepackage[colorlinks=true,allcolors=steelblue]{hyperref}
\usepackage{comment}
\usepackage{algorithm}
\usepackage{algorithmic}
\usepackage{scalerel}
\usepackage{bm}
\definecolor{steelblue}{RGB}{70,130,180}

\input{commands.tex}

\title{\LARGE \bf
An Online Learning Analysis of Minimax Adaptive Control
}

\author{Venkatraman Renganathan, Andrea Iannelli, and Anders Rantzer
\thanks{This project has received funding from the European Research Council (ERC) under the European Union’s Horizon 2020 research and innovation program under grant agreement No 834142 (Scalable Control). V. Renganathan and A. Rantzer are with the Department of Automatic Control, Lund University, Sweden and are also members of the ELLIIT Strategic Research Area in Lund University. A. Iannelli is with the Institute of Systems Theory and Automatic Control, University of Stuttgart, Germany. Emails: (venkatraman.renganathan,anders.rantzer)@control.lth.se, andrea.iannelli@ist.uni-stuttgart.de.}%
}

\else
\documentclass[10pt]{article}
\usepackage{graphicx}
\usepackage{amsmath}
\usepackage{optidef}
 
\usepackage{amsthm}
\usepackage{amssymb}
\usepackage{amsfonts}       
\usepackage{multicol}
\usepackage{cuted}
\usepackage{multirow}
\usepackage{caption, subcaption}
\newtheorem{theorem}{Theorem}
\newtheorem{lemma}{Lemma}

\newtheorem{definition}{Definition}

\usepackage{xcolor}
\usepackage{cite}
\usepackage[colorlinks=true,allcolors=steelblue]{hyperref}
\usepackage{comment}
\usepackage{algorithm}
\usepackage{algorithmic}
\usepackage{scalerel}
\usepackage{bm}
\usepackage{url}
\definecolor{steelblue}{RGB}{70,130,180}

\input{commands.tex}
\usepackage[preprint]{tmlr}

\title{An Online Learning Analysis of Minimax Adaptive Control}

\author{\name Venkatraman Renganathan  \email       venkatraman.renganathan@control.lth.se \\
      \addr Department of Automatic Control - LTH \\ Lund University, Sweden
      \AND
      \name Andrea Iannelli \email andrea.iannelli@ist.uni-stuttgart.de \\
      \addr Institute of Systems Theory and Automatic Control \\ University of Stuttgart, Germany.
      \AND
      \name Anders Rantzer \email anders.rantzer@control.lth.se \\
      \addr Department of Automatic Control - LTH \\ Lund University, Sweden
      }

\def\month{09}  
\def\year{2023} 
\def\openreview{\url{https://openreview.net/forum?id=XXXX}} 
\fi

\begin{document}

\maketitle
\thispagestyle{empty}
\pagestyle{empty}

\begin{abstract}

We present an online learning analysis of minimax adaptive control for the case where the uncertainty includes a finite set of linear dynamical systems. 
Precisely, for each system inside the uncertainty set, we define the model-based regret by comparing the state and input trajectories from the minimax adaptive controller against that of an optimal controller in hindsight that knows the true dynamics. We then define the total regret as the worst case model-based regret with respect to all models in the considered uncertainty set. We study how the total regret accumulates over time and its effect on the adaptation mechanism employed by the controller. Moreover, we investigate the effect of the disturbance on the growth of the regret over time and draw connections between robustness of the controller and the associated regret rate.

\end{abstract}

\section{Introduction}\label{sec_intro}
The interplay between machine learning, system identification and adaptive control has unveiled a fertile area of research which has the potential to answer some of the standing research questions in the field of learning-based control. 
Recent advances in online learning techniques have provided new perspectives on the design of algorithms where unknown systems can be controlled by acquiring knowledge through repeated interactions with the unknown environment \cite{hazan2022introduction}. This has close connections with adaptive control \cite{astromwittenmark} and in general with 
learning-based control techniques \cite{benosman2018model}. Minimax adaptive control is taken in this work as a prototypical example of the latter line of works to draw connections with regret, i.e. the performance metrics used in online learning. Design of minimax control for uncertain systems was investigated as early as in \cite{salmon1968minimax, didinsky1994minimax}. Subsequently, the design of minimax adaptive control was investigated for scalar systems with unknown input matrix sign in \cite{rantzer2020minimax}, for finite sets of linear systems in \cite{rantzer2021minimax, cederberg2022synthesis} and for the output feedback case in \cite{kjellqvist2021minimax}, respectively. There have been earlier works on robust adaptive control in \cite{chichka1995adaptive, yoneyama1997robust} where uncertainties in system dynamics were considered. Minimax adaptive control problems are generally challenging as obtaining exact $\ell_{2}$-gain bounds as explained in \cite{french2005p, vinnicombe2004examples, rantzer2021minimax} can be hard for multiple input multiple outputs systems with finite set of linear models and optimality can only be achieved if the exploration and exploitation trade-off is exactly captured. \\
The recent interest developed towards analyzing control algorithms for systems with unknown dynamics through the lens of regret analysis has the promise to enable a better understanding of this trade-off. There are quantities that are of interest but are unknown in advance to the online controller. We refer to such unknown entity as \emph{Quantity of Interest (QI)}. 
Lack of knowledge about a QI determines an accumulated cost, with respect to a control designed with perfect knowledge, which denotes the notion of regret. For instance, the growth of expected regret in linear quadratic control was investigated in \cite{jedra2022minimal} when matrices $(A, B)$ were unknown. Regret bounds have been investigated in \cite{boffi2021regret} for adaptive control problems in stochastic setting. This paper proposes an online learning analysis of minimax adaptive control of linear-time invariant systems featuring adversarial disturbance and a priori knowledge of a finite set of systems. One of the distinctive novelty is a new definition of regret, suitable for this setting in which the QIs are both the system dynamics and the exogenous disturbance. 
From an online learning perspective, efficient adaptive control algorithms are characterized by limiting the growth of regret over time. To quantify the regret, we usually require an optimal control policy (policy regret) or a sequence of best control actions (dynamic regret) available in hindsight as in \cite{goel2020regret}. Here, we propose studying the policy regret associated with the minimax adaptive controller by comparing it against the standard $\mathcal{H}_{\infty}$ control which knows the dynamics. \\
Recently, \cite{Karapetyan_CDC22} investigated the regret of robustness of an $\mathcal{H}_{\infty}$ controller (whose QI is just the adversarial disturbance) when compared to an oracle controller which has knowledge of the future disturbance trajectory. Also related is the work in \cite{hazan2020nonstochastic}, which investigated the regret analysis for the generic non-stochastic control problem and their system identification approach employed random inputs before controlling it using disturbance-based policy. Similarly, \cite{agarwal2019online} studied online control with adversarial disturbances and proposed a disturbance action control policy based efficient algorithm to obtain nearly tight regret bounds. On the contrary, our work looks at nonlinear adaptive state feedback policy which \emph{concurrently} controls the system under adversarial disturbance and implicitly learns the system dynamics. 
This gives rise to an interesting trade-off in the adversary strategy, whereby the worst-case disturbance is the one that delays
the learning process of the controller while minimizing the energy spent (which is penalized in the total cost). \\

\noindent \textbf{Contributions:} We provide a detailed analysis for the minimax adaptive control algorithm proposed in \cite{rantzer2021minimax}
with the aim to improve our understanding on the role of the adaptation mechanism and the adversary disturbance on the regret.
Since an explicit expression for the optimal minimax adaptive controller is not known, we apply our analysis to the candidate sub-optimal minimax adaptive control algorithm\footnote{The distinction between the optimal and the sub-optimal minimax adaptive control policies will be made clear at appropriate places.} developed in \cite{rantzer2021minimax, cederberg2022synthesis}. Specifically, the main contributions are: 
\begin{enumerate}
    \item Definition of the: \textsl{model-based regret} corresponding to a specific model in the uncertainty set characterizing the accumulated cost with respect to an optimal controller in hindsight which knows the true dynamics; \textsl{total regret} as the worst-case model-based regret corresponding to any model in the uncertainty set.
    \item Construction of an adversarial disturbance policy which provably prevents the minimax adaptive controller from learning the true dynamics (Theorem \ref{thm_confuse_disturb}). 
    \item Despite the possible difficulty in the identification of the true dynamics, we show that the minimax adaptive controller enjoys a sub-linear regret rate with respect to the best $\mathcal{H}_{\infty}$ controller in hindsight 
    (Theorem \ref{thm_regret_analysis}).
\end{enumerate}
The rest of the paper is organised as follows. The problem formulation is discussed in \S\ref{sec_minmax_adap_ctrl}. The online learning analysis is performed in \S\ref{sec_regret}, and some of its features are further elucidated through numerical simulation in \S\ref{sec_num_sim}. Finally, the main findings of the paper are summarized in \S\ref{sec_conclusion}.   

\noindent \textbf{Notation and Preliminaries.} 
The cardinality of the set $A$ is denoted by $\left | A \right \vert$. The set of real numbers, integers and the natural numbers are denoted by $\bbr, \bbz$, and $\bbn$ respectively. For a matrix $A \in \bbr^{n \times n}$, we denote its transpose and its trace by $A^{\top}$ and $\mathbf{Tr}(A)$ respectively. We denote by $\mathbb{S}^{n}$, the set of symmetric matrices in $\bbr^{n \times n}$. For $A \in \mathbb{S}^{n}$, we write $A \succ 0$ and $A \succeq 0$ to say that $A$ is positive definite and positive semi-definite, respectively. An identity matrix of dimension $n$ is denoted by $I_{n}$. Given $x \in \bbr^{n}, A \in \bbr^{n \times n}, B \in \bbr^{n \times n}$, the notations ${\left \| x \right \Vert}^{2}_{A}$ and ${\left \| B \right \Vert}^{2}_{A}$ mean $x^{\top} A x$ and $\mathbf{Tr}\left(B^{\top} A B \right)$ respectively. A signal $\{ x_k \}$ is said to be in $\ell_{2}$ space if it has finite energy meaning that $\sum^{\infty}_{k = 0} x^{2}_k < \infty$. For any time $T \in \mathbb{N}$, if the truncation of a signal $\{ x_k \}$ to the interval $[0, T]$ lies in the $\ell_{2}$ space, then the signal is said to be lying in the extended $\ell_{2}$ space denoted by $\ell_{2e}$.


\section{Problem Formulation Using Minimax Adaptive Control} \label{sec_minmax_adap_ctrl}
In this section we introduce the minimax adaptive control subject of our investigations through online learning.  

\subsection{Minimax adaptive control with finite set of linear systems} 
\noindent Consider the following discrete-time linear system
\begin{align}\label{eqn_system_dynamics}
    x_{k+1} = A x_{k} + B u_{k} + w_{k}, \quad k \in \mathbb{N},
\end{align}
where $x_{k} \in \bbr^{n}$ and $u_{k} \in \bbr^{m}$ denote the system states and control inputs, respectively, and the additive disturbance $w_{k} \in \bbr^{n}$ is assumed to be adversarial. 
The true system matrices $A \in \bbr^{n \times n}$ and $B \in \bbr^{n \times m}$ are unknown but assumed to belong to a set $\calm$ with $\left | \mathcal{M} \right \vert = \mathcal{F} \in \mathbb{N}$ defined such that 
$M_{i} := (A_{i}, B_{i}) \in \calm, i = 1, \dots, \mathcal{F}$, where all pairs are assumed throughout to be stabilizable. For instance, control of a discrete-time linearized inverted pendulum dynamics falls under the above setting when the pendulum length is uncertain. Generally, minimax adaptive control approach can be a suitable design solution when multiple systems who do not share common Lyapunov function need to be controlled by a single controller. Let us denote by $\Pi$ the set of admissible control policies
such that 
\begin{align} \label{eqn_control_policy}
    u_{k} = \pi_{k}\left( x_0, x_1, \dots, x_k, u_0, \dots, u_{k-1} \right),\quad \pi_{k} \in \Pi.
\end{align}
An optimal adaptive control policy should interact with the system in order to extract information about the unknown system matrices $A, B$ while also guaranteeing good performance and robustness to the adversarial disturbance. This can be achieved by optimizing the following minimax cost
\begin{align} \label{eqn_cost_function}
 \inf_{\pi \in \Pi} \underbrace{\sup_{w, A, B} \sum^{\infty}_{k=0} \left( c(x^{\pi}_{k}, u^{\pi}_{k}, Q, R) - \gamma^{2} {\left \| w_{k} \right \Vert}^{2} \right)}_{J_{\pi}(x_{0}, \gamma)}.
\end{align}
where $c(x^{\pi}, u^{\pi}, Q, R) := \norm{x^{\pi}}^{2}_{Q} + \norm{u^{\pi}}^{2}_{R}$ for given penalty matrices $Q \succ 0, R \succ 0$; $x^{\pi}$ denotes the evolution of the state of (\ref{eqn_system_dynamics}) starting from $x_0$ under the control input $u^{\pi}$ from the policy $\pi$; and $\gamma > 0$ quantifies the desired level of robustness to the external disturbance (higher $\gamma$ resulting in weaker robustness requirements). The optimal minimax control policy $\pi^{\dagger}$ and the associated cost are given by
\begin{align} \label{eqn_optimal_minimax_policy}
\pi^{\dagger} &:= \argmin_{\pi \in \Pi} J_{\pi}(x_{0}, \gamma), \quad J^{\dagger}(x_{0}, \gamma):= J_{\pi^{\dagger}}(x_{0}, \gamma)
\end{align}
and the resulting disturbance attenuation level achieved by the control policy $\pi^{\dagger}$ from disturbance to the regulated output $\zeta := \begin{bmatrix} x^{\top} & u^{\top} \end{bmatrix}^{\top}$ is denoted by $\gamma^{\dagger}$ and is defined as
\begin{align}
    \gamma^{\dagger} := \sqrt{\sup_{w^{\dagger} \neq 0} \frac{\sum^{\infty}_{k = 0} c(x^{\pi^{\dagger}}_{k}, u^{\pi^{\dagger}}_{k}, Q, R)}{\sum^{\infty}_{k = 0} \norm{w^{\dagger}_{k}}^{2}}}. 
\end{align}
This formulation provides a family of minimax control policy parameterized by $\gamma$, which are guaranteed to exist $\forall \gamma > \gamma^{\dagger}$. We cast the problem as a zero-sum dynamic game with the control policy $\pi$ being the minimizing player and the adversaries $(w, A, B)$ being the maximizing players \cite{rantzer2021minimax}. 
The solution boils down to solving a minimax dynamic programming problem, which is intractable in most cases. 
An approximate (i.e. sub-optimal) solution has been recently proposed in \cite{rantzer2021minimax, cederberg2022synthesis}, and this will be the subject of this study. The following lemma summarizes the main result of \cite{rantzer2021minimax}, i.e. an explicit expression for an adaptive controller satisfying a pre-specified $\ell_{2}$-gain bound from disturbance to error. 

\begin{lemma} \label{lemma_Anders_L4DC}  
Given a compact set of linear models $\mathcal{M}$, and positive definite penalty matrices $Q \in \mathbb{R}^{n \times n}, R \in \mathbb{R}^{m \times m}$, suppose that there exists $K_1, \dots, K_{\mathcal{F}} \in \mathbb{R}^{m \times n}$ and matrices $P_{ij} \in \mathbb{R}^{n \times n}$ with $0 \prec P_{ij} = P_{ji} \prec \gamma^{2} I$ such that 
\begin{align} \label{eqn_minimax_riccati}
    \norm{x}^{2}_{P_{il}} 
    \geq 
    \norm{x}^{2}_{Q} &+ \norm{K_l x}^{2}_{R} 
    - \gamma^{2} \norm{(\bar{A}_{il} - \bar{A}_{jl}) x/2}^{2} \nonumber \\
    &+ \norm{(\bar{A}_{il} + \bar{A}_{jl}) x/2}^{2}_{(P^{-1}_{ij} - \gamma^{-2} I)^{-1}}, 
\end{align}
where $\bar{A}_{il} = A_i - B_i K_l$ denotes the closed loop system matrix for $x \in \mathbb{R}^{n}$ with $i,j,l \in \{1,\dots,\mathcal{F} \}$. Then, the bound $J_{\bar{\pi}}(x_{0}, \gamma) \leq \max_{i,j} \norm{x_{0}}^{2}_{P_{ij}}$ is valid for the minimax adaptive control policy $\bar{\pi}$ defined by
\begin{subequations}\label{eqn_minimax_model_select_T}   
\begin{align}
u_{k} &= -K_{l_{k}} x_{k}, \quad \text{where}, \label{eqn_minimax_control_law} \\
l_{k} &:= \argmin_{i \in \{1,\dots,\mathcal{F} \}} \underbrace{\sum^{k-1}_{\tau=0} \norm{x_{\tau+1} - A_{i} x_{\tau} - B_{i} u_{\tau}}^{2}}_{:= \alpha_{i}}. \label{eqn_minimax_model_select}   
\end{align}
\end{subequations}
\end{lemma}
\noindent The controller defined in \eqref{eqn_minimax_control_law}, and denoted by $\bar{\pi}$ in the reminder, is sub-optimal compared to $\pi^{\dagger}$ (\ref{eqn_optimal_minimax_policy}), i.e. the associated $\ell_{2}$ gain is $\bar{\gamma} > \gamma^{\dagger}$. Further, the cost $J_{\bar{\pi}}(x_{0}, \gamma)$ is finite as long as $\gamma > \bar{\gamma}$. The control input \eqref{eqn_minimax_control_law} is nonlinear as it depends on all the past history, an approach based on least squares estimation from \cite{didinsky1994minimax}.

\subsection{Known Dynamics Case: Standard $\mathcal{H}_{\infty}$ Control}
When the system matrices $A, B$ are known, problem (\ref{eqn_cost_function}) reduces to the standard $\mathcal{H}_{\infty}$ control. That is, a control input $u = Kx, K \in \mathbb{R}^{m \times n}$ is sought such that it minimizes the 
$\mathcal{H}_{\infty}$ norm of the closed loop system from $d$ to $\zeta$ \begin{align} \label{eqn_tf_d_to_z}
    T_{d \rightarrow \zeta}[K](z) := \begin{bmatrix}
    I \\ K \end{bmatrix} (zI - A - BK)^{-1}.
\end{align}
where $T_{d \rightarrow \zeta}$ is related to the cost function in (\ref{eqn_cost_function}) by appropriate choice of matrices $Q, R$. Using this observation, we define for every system model $M_{i} := (A_{i}, B_{i}) \in \calm$, 
the associated $\mathcal{H}_{\infty}$ control policy $\pi^{\star}_{i} \in \Pi$, 
which can be found by solving the coupled Riccati equations below \cite{bacsar2008h}
\begin{align}
\mathbf{M}_{i} &= Q + A^{\top}_{i} \mathbf{M}_{i} \Lambda^{-1}_{i} A_{i}, \quad \mathbf{M}_{i} \prec (\gamma^{\star}_{i})^{2} I,  \\
\Lambda_{i} &= I + \left(B_{i} R^{-1} B^{\top}_{i} - \left(\gamma^{\star}_{i}\right)^{-2} I \right) \mathbf{M}_{i}.    
\end{align}
The dynamic game has an unique saddle point solution 
\begin{align}
u^{\pi^{\star}_{i}}_{k} &= \pi^{\star}_{i}(x_{k}) = - K^{\star}_{i} x_{k}, \quad \text{and} \label{eqn_hinfty_ctrl}\\
w^{\psi^{\star}_{i}}_{k} &= \psi^{\star}_{i}(x_{k}) = L^{\star}_{i} x_{k} \label{eqn_hinfty_dist},
\end{align}
where $K^{\star}_{i} = R^{-1} B^{\top}_{i} \mathbf{M}_{i} \Lambda^{-1}_{i} A_{i}$ and $L^{\star}_{i} = (\gamma^{\star}_{i})^{-2} \mathbf{M}_{i} \Lambda^{-1}_{i} A_{i}$. Here, $\psi^{\star}_{i}$ denotes the worst case adversarial disturbance policy and it is, like $\pi^{\star}_{i}$, a linear function of $x_{k}$. 
The quantity 
\begin{align}
    \gamma^{\star}_{i} := \sqrt{\sup_{w^{\psi^{\star}_{i}} \neq 0} \frac{\sum^{\infty}_{k=0} c\left(x^{\pi^{\star}_{i}}_{k}, u^{\pi^{\star}_{i}}_{k}, Q, R\right)}{ \sum^{\infty}_{k=0} \norm{w^{\psi^{\star}_{i}}_{k}}^{2}_{2}}}
\end{align}
denotes the corresponding worst-case $\ell_{2}$ gain from the disturbance to the regulated output for the model $M_{i}, i \in \mathcal{M}$. 

\section{Regret of Minimax Adaptive Control} \label{sec_regret}
Regret analysis compares the performance of an online algorithm that takes decisions in the presence of uncertainty with respect to a clairvoyant policy with hindsight knowledge. For this reason, it is used here in order to better understand the performance achieved when controlling the system \eqref{eqn_system_dynamics} using minimax adaptive control algorithm. 
\begin{definition}
Regret of an online control algorithm $\mathcal{A}$ operating in the presence of uncertainty is defined as the additional cost incurred by the algorithm $\mathcal{A}$ in comparison to an optimal controller in hindsight that operates by knowing the uncertainty.
\end{definition}

We choose here the $\mathcal{H}_{\infty}$ controller associated with the true system as the optimal policy in hindsight. 
Note that $\forall i \in 1, \dots, \mathcal{F}$, $J_{\pi^{\star}_{i}}(x_{0}, \gamma) < J^{\dagger}(x_{0}, \gamma)$ as the minimax adaptive control policy $\pi^{\dagger}$ can never do better than the $\mathcal{H}_{\infty}$ policy $\pi^{\star}_{i}$ of the corresponding true system.
It is possible to use a different policy other than the $\mathcal{H}_{\infty}$ policy for the comparison. One could compare against a control policy that solves the linear quadratic problem with known disturbance but the true $(A,B) \in \mathcal{M}$ being unknown. However, to the best of our knowledge, there is no \emph{causal} solution for the optimal control policy to that problem. In principle, the optimal policy in hindsight should know apriori about any of the QIs that minimax does not know and also have a closed form causal solution. 

The study of the minimax adaptive control problem through online learning is divided in three steps: investigation of adversarial disturbance strategies that can lead to performance deterioration of the policy $\bar{\pi}$; definition of suitable notions of regret for this problem; investigation of the regret properties of the policy $\bar{\pi}$. We do not advocate the regret as a metric to measure the robustness of a control policy. Rather, we suggest to use the regret as a tool to identify areas of improvement of an online control policy by comparing it against multiple optimal policies in hindsight. Regret analysis could also give insights for the online control design to foresee and counteract against several possible strategies of adversaries trying to worsen its performance. One such possible strategy of an adversary with respect to the  policy $\bar{\pi}$ is illustrated below.

\subsection{Adversarial disturbance strategies for minimax control} \label{subsec_adverse_policy_construction}
The key adaptive mechanism of policy $\bar \pi$ in \eqref{eqn_minimax_model_select_T} can be interpreted as an implicit identification of the underlying plant. 
It is then natural to ask whether this is provably able to eventually converge to the correct estimate for the system. The following theorem gives a negative answer by constructing an adversarial disturbance strategy preventing the controller from optimally controlling the true system. 
\begin{theorem} \label{thm_confuse_disturb}
Given a compact set of models $\mathcal{M}$ with $\left | \mathcal{M} \right \vert = \mathcal{F}$ including the true model of the system \eqref{eqn_system_dynamics}, consider the policy $\bar{\pi}$ given by \eqref{eqn_minimax_model_select_T}. Let $j \in \{1,\dots,\mathcal{F}\}$ denote the index of the true model unknown to the policy $\bar{\pi}$. Then, $\forall k \in \mathbb{N}$, $\exists \theta_{f, k} \in \mathbb{R}$, $f = 1, \dots, \mathcal{F}$ such that the disturbance given by
\begin{align} \label{eqn_minimax_confusing_disturbance}
    w_{k} = \sum^{\mathcal{F}}_{f = 1}  \theta_{f, k} (A_{f} x_{k} + B_{f} u_{k}),
\end{align}
lets $\bar{\pi}$ to determine a minimizer $l_{k} \neq j$ in \eqref{eqn_minimax_model_select}.
\end{theorem}
\begin{proof}
Recall from \eqref{eqn_minimax_model_select} that when $i = j$, we simply get $\alpha_{j} = \sum^{k-1}_{\tau=0} \norm{w_{\tau}}^{2}$. For other cases when $i \neq j$, we expand $\alpha_{i}$ using the $w_{k}$ given by \eqref{eqn_minimax_confusing_disturbance} to get
\begin{align} \label{eqn_alpha_i_neq_jk}
    \alpha_{i} 
    &= \sum^{k-1}_{\tau=0} \norm{ v^{(i)}_{\tau} + \sum^{\mathcal{F}}_{f = 1, f \neq j} \theta_{f, k} (A_{f} x_{k} + B_{f} u_{k}) }^{2}, 
\end{align}
with $v^{(i)}_{\tau} = (\theta_{j, k} A_{j} - A_{i}) x_{\tau} + (\theta_{j, k} B_{j} - B_{i}) u_{\tau}$. Then, the disturbance can let the controller choose $l_{k} = i$ as per \eqref{eqn_minimax_model_select} deviating from the true value of $j$ through the appropriate selection of the constants $\{\theta_{f,k}\}^{\mathcal{F}}_{f=1}$  such that $\alpha_{i} < \alpha_{j}$. One simple choice would be to choose $\theta_{j,k} = -1, \theta_{i,k} = 1$ and $\{\theta_{f, k}\}^{\mathcal{F}}_{f=1, f \neq i, f \neq j} = 0$ at time $k$ such that $\alpha_{i} = 0$ in \eqref{eqn_alpha_i_neq_jk}. Such a disturbance strategy would let the controller choose $l_{k} = i$ rather than $j$. Note that the adversary has the freedom to make $\alpha_{i} = 0$ for its own choice of $i \in \{1,\dots,\mathcal{F}\}, i \neq j$ at any time step $k$ using $\{\theta_{f,k}\}^{\mathcal{F}}_{f=1}$. 
\end{proof}

\noindent \textbf{Remarks:} Disturbances with smaller magnitudes maximise the cost given in \eqref{eqn_cost_function}. Though, the disturbance given by \eqref{eqn_minimax_confusing_disturbance} can make the learning hard for the controller, it need not have a smaller magnitude for a given $\gamma > 0$ and $\{\theta_{f,k}\}^{\mathcal{F}}_{f=1}$, and hence it may \emph{not} lead to the worse cost. Further, for certain range of $\{\theta_{f,k}\}^{\mathcal{F}}_{f=1}$, the associated closed loop system may turn out to be unstable. The negative result formulated in Theorem \ref{thm_confuse_disturb} justifies further analysis on the sub-optimality faced by the minimax adaptive controller, which is studied in the next sections through the concept of regret.

\subsection{Regret Definitions}
Note that each model $(A_{i}, B_{i}) \in \mathcal{M}$ suffers different regret when compared against the optimal $\mathcal{H}_{\infty}$ controller in hindsight. Hence, we quantify the regret of each model in the set $\mathcal{M}$ in the following definition.
\begin{definition} Given a model $M_{i} := (A_{i}, B_{i}) \in \calm, i \in \{1, \dots, \mathcal{F}\}$, we define the model-based regret of the minimax adaptive control policy $\pi^{\dagger} \in \Pi$ with respect to the optimal control policy $\pi^{\star}_{i}$ for $\gamma \geq \gamma^{\dagger} > \gamma^{\star}_{i}$ and time $T \in \mathbb{N}$ as 
\begin{align} \label{eqn_model_regret_def_cost}
    \calr(\pi^{\dagger}, \pi^{\star}_{i}, T) 
    = \sup_{w \in \ell_{2e}} \sum^{T}_{k=0} d_{k} \left(\pi^{\dagger}, \pi^{\star}_{i} \right), \quad \text{where}, 
    \end{align}
\begin{equation*}
    d_{k} \left(\pi^{\dagger}, \pi^{\star}_{i} \right) 
    := \norm{x^{\pi^{\dagger}}_{k} - x^{\pi^{\star}_{i}}_{k}}^{2}_{Q} + \norm{u^{\pi^{\dagger}}_{k} - u^{\pi^{\star}_{i}}_{k}}^{2}_{R}. 
\end{equation*}
\end{definition}
Any disturbance that is not in the $\ell_{2e}$ space will result in diverging states. We note that the choice of regret metric is not conventional, as the standard approach would be to define it as difference of costs, that is,
\begin{align} \label{eqn_model_regret_def_cost_difference}
    \bar{\calr}(\pi^{\dagger}, \pi^{\star}_{i}, T) 
    &= \sup_{w \in \ell_{2e}} \sum^{T}_{k=0} \bar{d}_{k} \left(\pi^{\dagger}, \pi^{\star}_{i} \right),  
\end{align}
\begin{equation*}
    \bar{d}_{k} \left(\pi^{\dagger}, \pi^{\star}_{i} \right) := c\left(x^{\pi^{\dagger}}_{k}, u^{\pi^{\dagger}}_{k}, Q, R \right) - c\left(x^{\pi^{\star}_{i}}_{k}, u^{\pi^{\star}_{i}}_{k}, Q, R \right)
\end{equation*}
While \eqref{eqn_model_regret_def_cost_difference} captures how close the systems controlled by the minimax adaptive controller and the optimal $\mathcal{H}_{\infty}$ controller in hindsight are in terms of the performance, 
it does not provide information on how close the two state and inputs trajectories are. Further, \eqref{eqn_model_regret_def_cost_difference} cannot account for the direction of the control input being applied to the system. For these reasons, we propose to use \eqref{eqn_model_regret_def_cost} as the definition of model-based regret in this work. Note that the model-based regret in \eqref{eqn_model_regret_def_cost} is a function of the chosen level of robustness $\gamma$ because this parameter affects the two policies $\pi^{\dagger}$ and $\pi^{\star}_{i}$ (this dependence is omitted for the sake of clarity). The regret is defined for $\gamma \geq \gamma^{\dagger} > \gamma^{\star}_{i}$ to ensure that a fair comparison is made between the resulting trajectories from controllers that share the same level of disturbance attenuation capabilities. To compute \eqref{eqn_model_regret_def_cost}, we need to characterize the trajectories of the system $x^{\pi^{\dagger}}_{k}$ and $x^{\pi^{\star}_{i}}_{k}$ given by \eqref{eqn_system_dynamics} under the same sequence of adversarial disturbance inputs affecting the system using the control policies $\pi^{\dagger}$ and $\pi^{\star}_{i}$ respectively. This naturally leads us to investigate what would be the worst-case model-based regret corresponding to any arbitrary model $M_{i} \in \mathcal{M}$, i.e., the total regret.  

\begin{definition} 
The total regret of the minimax adaptive controller is defined as
\begin{align} \label{eqn_total_regret_def_cost}
    \mathfrak{R}(\pi^{\dagger}, T) 
    := \max_{i \in \{1, \dots, \mathcal{F}\}} \calr(\pi^{\dagger}, \pi^{\star}_{i}, T).
\end{align}
\end{definition}
While comparing policies, it is important to compare their disturbance attenuation levels too. Sub-optimality gap indicates a room for improvement in terms of the robustness. Since, minimax adaptive controller can never match the $\mathcal{H}_{\infty}$ controller, the difference in their disturbance attenuation level is referred as the \emph{model-based sub-optimality gap}. 
\begin{definition}
Given a model $(A_{i}, B_{i}) \in \calm, i \in \{1, \dots, \mathcal{F}\}$, the model-based sub-optimality gap of the minimax adaptive control policy $\pi^{\dagger}$ is defined as 
\begin{align} \label{eqn_model_based_subopt_gap}
    {\calo}(\pi^{\dagger}, \pi^{\star}_{i}) &:= \gamma^{\dagger} - \gamma^{\star}_{i}. 
\end{align}
\end{definition}
The model-based sub-optimality gap satisfies by definition 
${\calo}(\pi^{\dagger}, \pi^{\star}_{i}) \geq 0$ and characterizes how the lack of knowledge about the QIs results in a worst disturbance attentuation level of the minimax adaptive controller (or reduction in robust performance). In a similar spirit to the definition of total regret, we define below 
the minimal and the maximal sub-optimality gaps, which are by definition both non-negative.
\begin{definition} The minimal sub-optimality gap and the maximal sub-optimality gap of the minimax adaptive control policy $\pi^{\dagger}$ are respectively defined as 
\begin{align}
    \underline{\calo}(\pi^{\dagger}) &:= \gamma^{\dagger} - \max_{i \in \{1, \dots, \mathcal{F}\}} \gamma^{\star}_{i}, \quad \text{and} \\
    \overline{\calo}(\pi^{\dagger}) &:= \gamma^{\dagger} - \min_{i \in \{1, \dots, \mathcal{F}\}} \gamma^{\star}_{i}. 
\end{align}
\end{definition}

\subsection{Study of Minimax Adaptive Control Regret} \label{sec_regret_analysis}
The following theorem establishes the asymptotic behaviour of the total regret associated with the minimax adaptive control policy $\mathfrak{R}(\bar{\pi}^{\dagger}, T)$. 
\begin{theorem} \label{thm_regret_analysis}
Consider the uncertain linear dynamical system given by \eqref{eqn_system_dynamics} with the uncertainty described by $\calm$. If the disturbance signal is in $\ell_{2}$ space, then the associated total regret (\ref{eqn_total_regret_def_cost}) is sub-linear, i.e.
\begin{align} \label{eqn_regret_growth_cdtn}
    \lim_{T\rightarrow \infty} \frac{\mathfrak{R}(\bar{\pi}^{\dagger}, T)}{T} = 0.
\end{align}
\end{theorem}
\begin{proof}
Recall that both minimax adaptive control policy $\bar{\pi}$ given by \eqref{eqn_minimax_control_law} and $\mathcal{H}_{\infty}$ control policy given by \eqref{eqn_hinfty_ctrl} are stabilising (with exponential decay of states and controls) for any adversarial disturbance in $\ell_{2}$ space. That is, given any disturbance signal $w_{k}$ in $\ell_{2}$ space for plant model $i \in \{1,\dots,\mathcal{F}\}$, we have 
\begin{align} \label{eqn_state_exp_decay}
\lim_{k \rightarrow \infty} \norm{x^{\bar{\pi}^{\dagger}}_{k}}^{2}_{Q} = 0, \quad \text{and} \quad 
\lim_{k \rightarrow \infty} \norm{x^{\pi^{\star}_{i}}_{k}}^{2}_{Q} = 0.
\end{align}
Then, this means that $\lim_{k \rightarrow \infty} \norm{x^{\bar{\pi}^{\dagger}}_{k} - x^{\pi^{\star}_{i}}_{k}}^{2}_{Q} = 0$. Since at any time $k \in \mathbb{N}$ both minimax adaptive control input $u^{\bar{\pi}^{\dagger}}_{k}$ given by \eqref{eqn_minimax_control_law} and the $\mathcal{H}_{\infty}$ control input $u^{\pi^{\star}_{i}}_{k}$ given by \eqref{eqn_hinfty_ctrl} are functions of the states $x^{\bar{\pi}^{\dagger}}_{k}$ and $x^{\pi^{\star}_{i}}_{k}$ respectively that decay to zero asymptotically, we infer that
\begin{align} \label{eqn_input_exp_decay}
\lim_{k \rightarrow \infty} \norm{u^{\bar{\pi}^{\dagger}}_{k}}^{2}_{R} = 0, \quad \text{and} \quad 
\lim_{k \rightarrow \infty} \norm{u^{\pi^{\star}_{i}}_{k}}^{2}_{R} = 0.
\end{align}
Then, this means that $\lim_{k \rightarrow \infty} \norm{u^{\bar{\pi}^{\dagger}}_{k} - u^{\pi^{\star}_{i}}_{k}}^{2}_{R} = 0$. Therefore, the difference term decays as well to zero meaning that $\lim_{k \rightarrow \infty} d_{k}(\bar{\pi}^{\dagger}, \pi^{\star}_{i}) = 0$. Hence, the result follows.
\end{proof}
An insight gathered from the proof is that stability of the policy implies certain regret properties. This has connections with recent findings in \cite{Karapetyan_IFAC23} which studied the relationship between stability and regret for disturbances in $\ell_{\infty}$ space.

\section{Numerical Simulation}
\label{sec_num_sim}
In this section, we exemplify our analysis using a linear dynamical system with a model uncertainty consisting of four different linear models. 
\begin{figure*}
  \centering
  \begin{subfigure}{.33\linewidth}
    \centering
    \includegraphics[width=\textwidth]{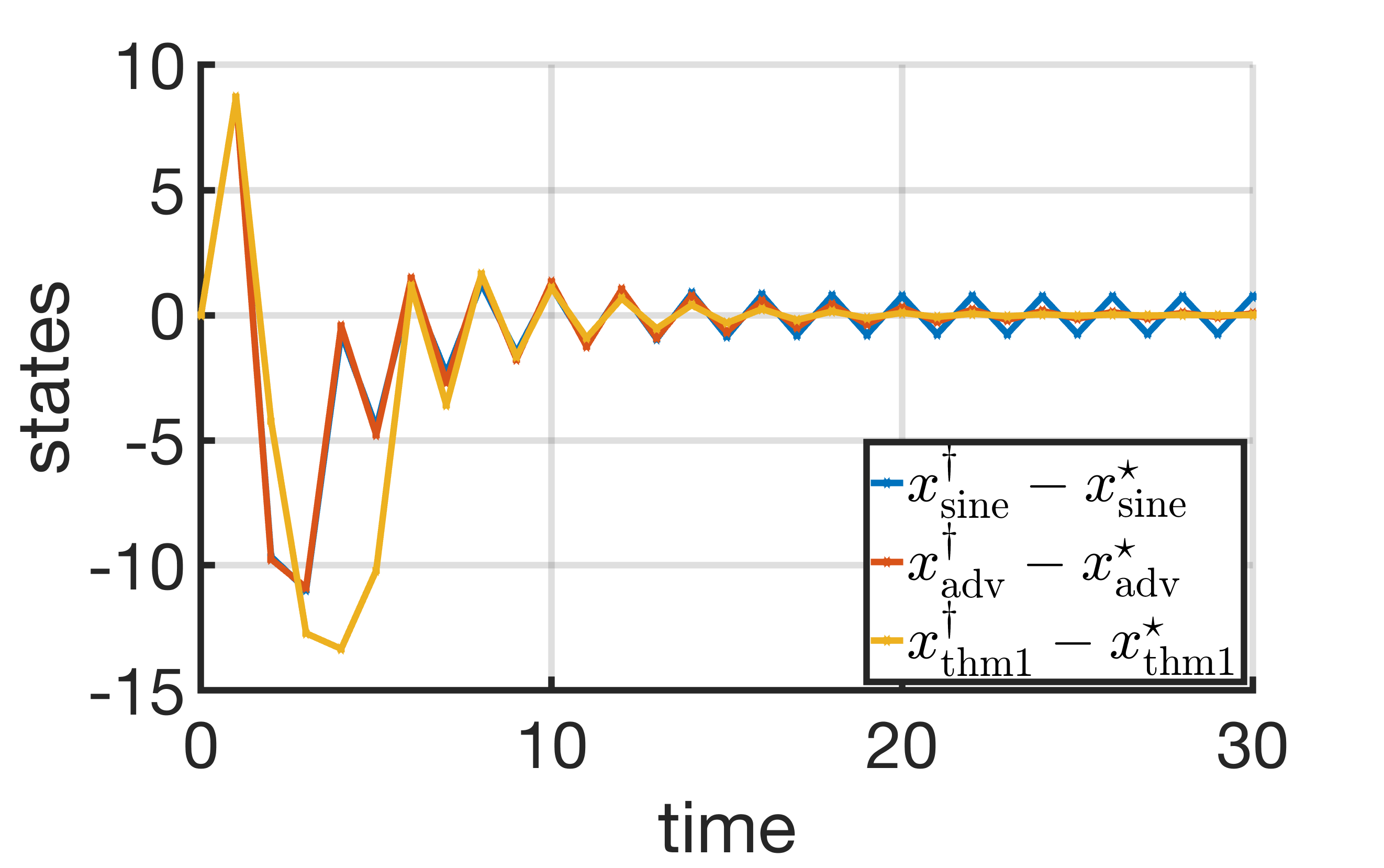}
    \caption{Minimax \& $\mathcal{H}_{\infty}$ States: $x^{\bar{\pi}^{\dagger}}_{k}$}
  \end{subfigure}%
  \begin{subfigure}{.33\linewidth}
    \centering
    \includegraphics[width=\textwidth]{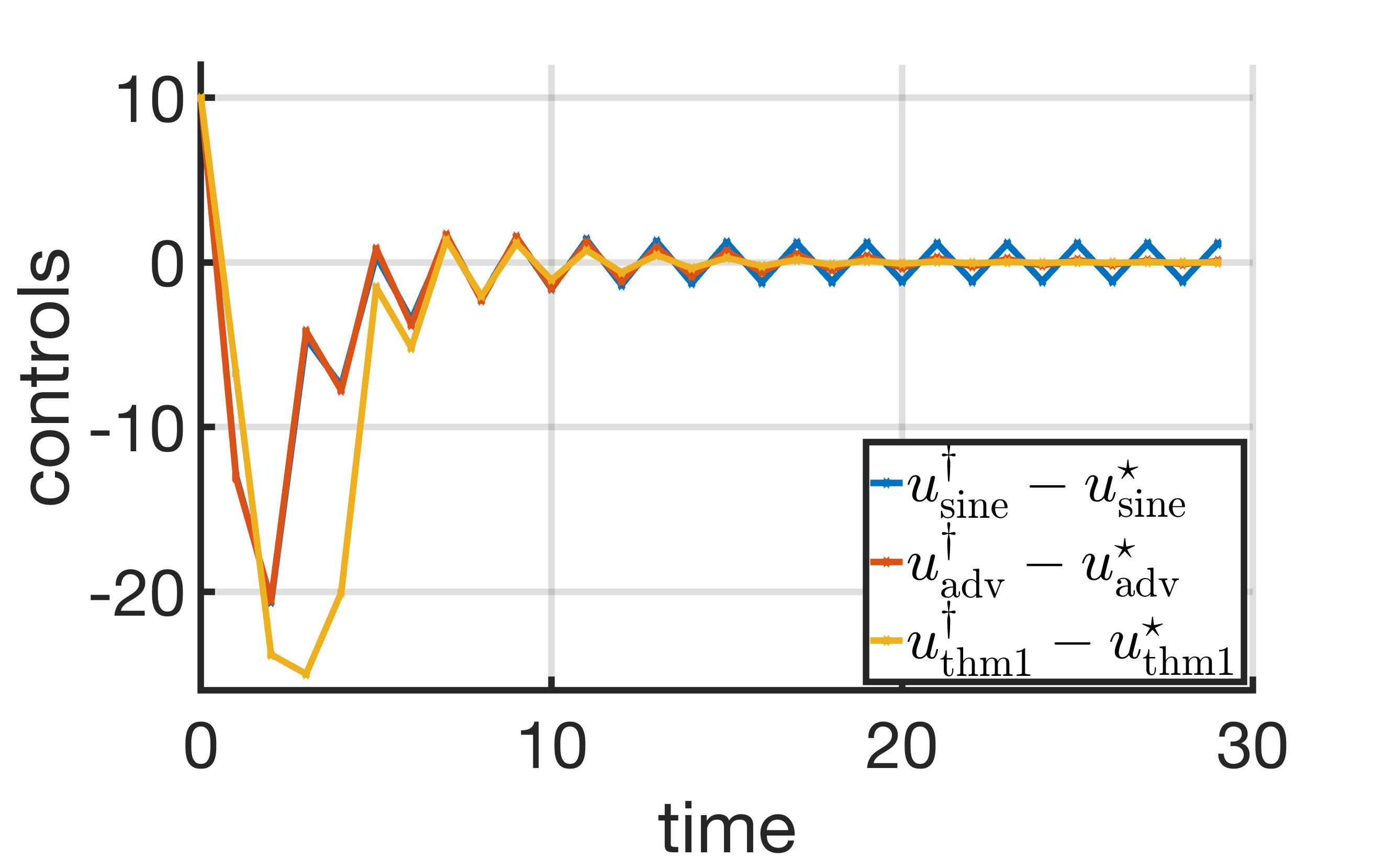}
    \caption{Minimax \& $\mathcal{H}_{\infty}$ Controls: $x^{\pi^{\star}_{2}}_{k}$}
  \end{subfigure}%
  \begin{subfigure}{.33\linewidth}
    \centering
    \includegraphics[width=\textwidth]{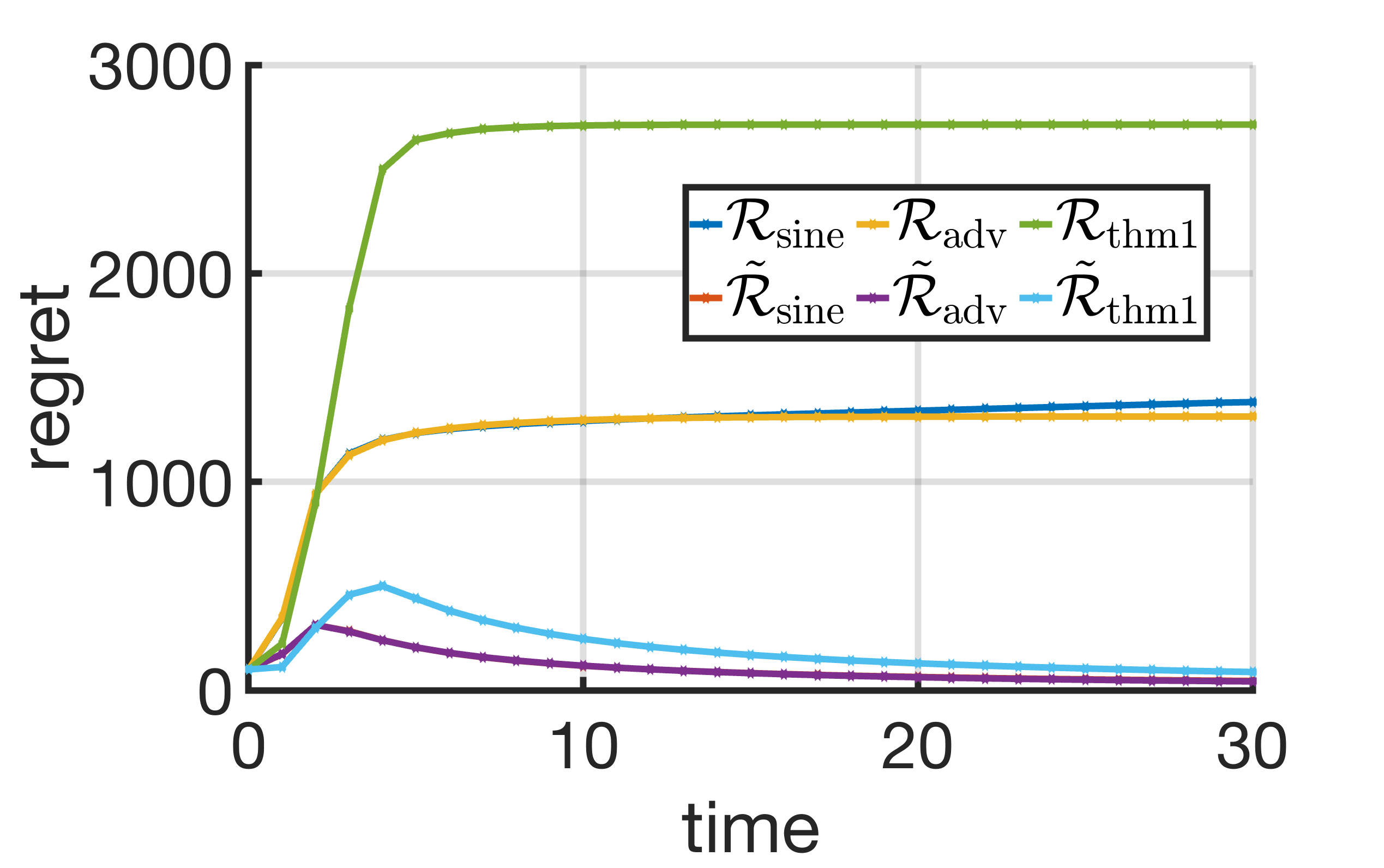}
    \caption{Regret scaling vs time}
  \end{subfigure}
 \caption{Simulation results with states and controls from the minimax adaptive controller and the $\mathcal{H}_{\infty}$ controller are plotted here along with the corresponding regret scaling over time. Only the first state of the system \eqref{eqn_random_dyn_sys} is plotted for the demonstration purpose. Note that the quantities $\mathcal{R}(\bar{\pi}^{\dagger}, \pi^{\star}_{2}, T)$ and $\frac{\mathcal{R}(\bar{\pi}^{\dagger}, \pi^{\star}_{2}, T)}{T}$ are abbreviated as $\mathcal{R}$ and $\tilde{\mathcal{R}}$ respectively. The text in the subscript of quantities in all sub-plots denotes the type of disturbance being used.}

  \label{fig_0}
\end{figure*}

\subsection{Problem Setup}
We consider the following numerical example of a linear dynamical system with four possible models. The state and control penalty matrices were $Q = I_{3}, R = 1$ and $T = 50$. We simulated the system using the minimax adaptive controller and the $\mathcal{H}_{\infty}$ controller available in hindsight separately when the pair $(A_{2}, B_{2})$ (corresponds to $j = 2$ as per Theorem \ref{thm_confuse_disturb}) was the true model.
\begin{equation} \label{eqn_random_dyn_sys}
\footnotesize
\begin{aligned}
A_{1} &= \begin{bmatrix}
1.908 &  0.853 & 0.633 \\
0.853 &  0.142 & 0.645 \\
0.633 &  0.645 & 0.018
\end{bmatrix}, 
A_{2} = \begin{bmatrix}
0.060 & 0.335 & 0.809 \\
0.335 & 0.017 & 1.507 \\
0.809 & 1.507 & 0.873
\end{bmatrix}, \\
A_{3} &= \begin{bmatrix}
0.182 &   1.435 &   0.730 \\
1.435 &   1.714 &   1.183 \\
0.730 &   1.183 &   0.452
\end{bmatrix}, 
A_{4} = \begin{bmatrix}
0.922  &  0.800  &  1.350 \\
0.800  &  1.431  &  1.462 \\
1.350  &  1.462  &  0.786
\end{bmatrix}, \\
B_{1} &= \begin{bmatrix}
1.830 & 1.285 & 0.002
\end{bmatrix}^{\top},
B_{2} = \begin{bmatrix}
0.873 & 0.098 & 0.099
\end{bmatrix}^{\top}, \\
B_{3} &= \begin{bmatrix}
1.073 & 1.524 & 0.695
\end{bmatrix}^{\top},
B_{4} = \begin{bmatrix}
0.358 & 1.266 & 1.248
\end{bmatrix}^{\top}.
\end{aligned}
\end{equation}
Three different disturbances constructions were used 
\begin{enumerate}
    \item worst case disturbance signal obtained from the dynamic game based $\mathcal{H}_{\infty}$ approach given by \eqref{eqn_hinfty_dist}. 
    \item sinusoidal disturbance with unit amplitude and its frequency being selected as the frequency where the $\mathcal{H}_{\infty}$ norm of $T_{d \rightarrow \zeta}[K](z)$ given by \eqref{eqn_tf_d_to_z} was maximum.
    \item the disturbance given by \eqref{eqn_minimax_confusing_disturbance} used in the proof of Theorem \ref{thm_confuse_disturb} with $i = 3$ and tuned so that the controller always choose the optimal controller for $(A_3, B_3)$. 
\end{enumerate}
The gains for the sub-optimal minimax adaptive controller were calculated using the method from \cite{cederberg2022synthesis} which is an improved version of Theorem 3 in \cite{rantzer2021minimax} and we used the Yalmip toolbox with the MOSEK solver to solve the associated convex optimization problem with linear matrix inequality constraints. 
The code corresponding to the figures given in the paper is made publicly available at \url{https://github.com/venkatramanrenganathan/minimaxadaptivecontrolregret}. 

\subsection{Results \& Discussions}
The system dynamics with $\{(A_{i}, B_{i})\}^{4}_{i=1}$ given by \eqref{eqn_random_dyn_sys} was solved for the minimax adaptive control policy $\bar{\pi}$ to get $\bar{\gamma}^{\dagger} = 31.0086$. Then, the corresponding $\mathcal{H}_{\infty}$ controller was obtained using $\gamma = \bar{\gamma}^{\dagger}$. The optimal $\ell_{2}$ gains, namely $\{\gamma^{\star}_{i}\}^{4}_{i=1}$ corresponding to the optimal $\mathcal{H}_{\infty}$ controller for the plants $\{(A_{i}, B_{i})\}^{4}_{i=1}$ were $1.266, 4.544, 2.913, 2.298$  respectively. The model-based sub-optimality gaps were found using \eqref{eqn_model_based_subopt_gap} as ${\calo}(\bar{\pi}^{\dagger}, \pi^{\star}_{1}) = 29.7426$, ${\calo}(\bar{\pi}^{\dagger}, \pi^{\star}_{2}) = 26.4646$,
${\calo}(\bar{\pi}^{\dagger}, \pi^{\star}_{3}) = 28.0956$, and
${\calo}(\bar{\pi}^{\dagger}, \pi^{\star}_{4}) = 28.7106$. Further, the minimal and maximal sub-optimality gaps were $\underline{\calo}(\bar{\pi}^{\dagger}) = 26.4646$ and $\overline{\calo}(\bar{\pi}^{\dagger}) = 29.7426$ respectively. 

The results of simulating system with minimax adaptive controller and $\mathcal{H}_{\infty}$ controller with 
three different disturbance strategies are shown in Figure~\ref{fig_0}. The sub-figures ~\ref{fig_0}(a), and ~\ref{fig_0}(b) depict the difference of states and inputs respectively from the minimax adaptive controller and the $\mathcal{H}_{\infty}$ controller and precisely these are the main contributing factors of regret as per \eqref{eqn_total_regret_def_cost}. The regret quantities $\mathcal{R}(\bar{\pi}^{\dagger}, \pi^{\star}_{2}, T)$ and $\frac{\mathcal{R}(\bar{\pi}^{\dagger}, \pi^{\star}_{2}, T)}{T}$ are abbreviated as $\mathcal{R}$ and $\tilde{\mathcal{R}}$ respectively are plotted in the sub-figure ~\ref{fig_0}(c). When adversarial disturbance constructed from the worst case disturbance policy given by \eqref{eqn_hinfty_dist} was used, the associated regret was bounded as seen in sub-figure ~\ref{fig_0}(c). Moreover, the shown results fulfils \eqref{eqn_regret_growth_cdtn} as both the control policies were stabilising and the disturbance was regulated to zero as it was a linear function of the system states (as per \eqref{eqn_hinfty_dist}) which decayed in exponential time. When a sinusoidal disturbance with unit amplitude was employed with its frequency being selected as the frequency where the $\mathcal{H}_{\infty}$ norm of $T_{d \rightarrow \zeta}[K](z)$ given by \eqref{eqn_tf_d_to_z} was maximum (for $(A_{2}, B_{2})$ this happens at $\pi \, \mathrm{rad/s}$), the regret was not bounded anymore as the sinusoidal disturbance does not belong to the $\ell_{2}$ space. However, the ratio namely $\mathcal{R}(\bar{\pi}^{\dagger}, \pi^{\star}_{2}, T)/T$ went to zero asymptotically as the terms contributing to the regret namely the differences of states and controls remained small and did grow slower than linearly. When the disturbance given by \eqref{eqn_minimax_confusing_disturbance} was used, it ensured that the $l_{k}$ value chosen by the minimax adaptive control policy $\bar{\pi}^{\dagger}$ according to \eqref{eqn_minimax_model_select} was never equal to $j = 2, \forall k \in [0,T]$. 
Even though the ratio $\mathcal{R}(\bar{\pi}^{\dagger}, \pi^{\star}_{2}, T)/T$ went to zero asymptotically as disturbance was still a function of exponentially decaying states, it can be observed that the disturbance signal had a larger magnitude than the one from \eqref{eqn_hinfty_dist}. This shows that the disturbance that hardens the learning process need not necessarily worsen the performance as measured by \eqref{eqn_cost_function} because it results in a decrease of the total cost.


\section{Conclusion} \label{sec_conclusion}
An online learning-inspired analysis for a recently proposed solution for a class of minimax adaptive control problems has been presented. Model-based regret and total regret for the minimax adaptive control policy were defined by comparing the state and input trajectories against that of the optimal $\mathcal{H}_{\infty}$ controller in hindsight (i.e. having knowledge of the true system dynamics). One of the highlights of the analysis is that the total regret is sub-linear for exogenous disturbances in the $\ell_{2}$ space, confirming links between system theoretic properties and regret for control systems. Future research will seek to  
characterize transient properties of regret, and their connections with the exploration-exploitation trade-off inherently captured by the
minimax adaptive control algorithms. Starting from the definitions of regret proposed here, designing regret-optimal adaptive controllers that lower the conservatism of minimax solutions is also an important research question lying ahead. 



\section*{Acknowledgment}
\inArxiv{
This project has received funding from the European Research Council (ERC) under the European Union’s Horizon 2020 research and innovation program under grant agreement No 834142 (Scalable Control). V. Renganathan and A. Rantzer are members of the ELLIIT Strategic Research Area in Lund University.
}
The authors are thankful to Daniel Cederberg at Linköping University for providing us with the code and to Olle Kjellqvist at Lund University for his insightful comments.


\inConf{
\bibliographystyle{IEEEtran}
\bibliography{references}
}
\inArxiv{
\bibliographystyle{tmlr}
\bibliography{references}
}


\end{document}

%% file: commands.tex
\def\bbn{\mathbb N}
\def\bbz{\mathbb Z}
\def\bbr{\mathbb R}

\def\calm{\mathcal M}

\def\calr{\mathcal R}

\def\calo{\mathcal O}

\DeclareMathOperator*{\argmin}{argmin}

\newcommand{\norm}[1]{\left\lVert {#1} \right\rVert}
